\documentclass{sig-alternate-per}

\newcommand{\ignore}[1]{}
\def\opt{\mathsf{OPT}}

\allowdisplaybreaks[4]
\usepackage{tikz}
\usepackage{pgfplots}
\definecolor{bblue}{HTML}{4F81BD}
\definecolor{rred}{HTML}{C0504D}
\definecolor{ggreen}{HTML}{9BBB59}
\definecolor{ppurple}{HTML}{9F4C7C}

\usetikzlibrary{calc}

\newtheorem{theorem}{Theorem}
\newtheorem{corollary}[theorem]{Corollary}
\newtheorem{lemma}[theorem]{Lemma}
\newtheorem{remark}[theorem]{Remark}
\newenvironment{proof}{ \textbf{Proof:} }{ \hfill $\Box$}

\def\bb0{{\mathbb{0}}}

\def\bb{{\mathbf{b}}}

\def\bp{{\mathbf{p}}}

\def\b0{{\mathbf{0}}}

\def\opt{\mathsf{OPT}}

\def\b1{{\mathbf{1}}}

\def\cA{\mathcal{A}}

\def\cJ{\mathcal{J}}

\def\sf0{{\mathsf{0}}}

\def\nn{\nonumber}

\newlength{\figurewidth}

\newcounter{one}
\newcounter{two}
\begin{document}
\title{Speed Scaling with Multiple Servers Under A  Sum Power Constraint}
\numberofauthors{2} 
\author{
\alignauthor
Rahul Vaze\thanks{We acknowledge support of the Department of Atomic Energy, Government of India, under project no. RTI4001}\\
       \affaddr{School of Technology and Computer Science}\\
       \affaddr{Tata Institute of Fundamental Research}\\
       \email{rahul.vaze@gmail.com}
\alignauthor
Jayakrishnan Nair\\
       \affaddr{Department of Electrical Engineering}\\
       \affaddr{IIT Bombay}\\
       \email{jayakrishnan.nair@iitb.ac.in}
}


%
%

\maketitle

\begin{abstract}
The problem of scheduling jobs and choosing their respective speeds with multiple servers under a sum power constraint  to minimize the flow time + energy is considered. This problem is a generalization of the flow time minimization problem with multiple unit-speed servers, when jobs can be parallelized, however,   
with a sub-linear, concave speedup function $k^{1/\alpha}, \alpha>1$ when allocated $k$ servers, i.e., jobs experience diminishing returns from being allocated additional servers. When all jobs are available at time $0$, we show that a very simple algorithm EQUI, that processes all available jobs at the same speed is $\left(2-\frac{1}{\alpha}\right) \frac{2}{\left(1-\left(\frac{1}{\alpha}\right)\right)}$-competitive, while in the general case, when jobs arrive over time, an LCFS based algorithm is shown to have a constant (dependent only on $\alpha$) competitive ratio.
\end{abstract}

\section{Introduction}
Scheduling jobs with multiple servers to minimize the sum of their response times (called the flow time) is an important practical problem, and finding optimal algorithms remains challenging. 
An added feature in modern servers is their ability to work at different speeds. This paradigm is called {\it speed scaling} \cite{bansal2009speed,devanur2017primal, vaze2020multiple, vaze2020network}, where one or more servers with tuneable speed are available, and operating any server at speed $s$ consumes energy at rate $P(s)$, a non-decreasing convex function of $s$. 
With speed scaling, the problem is to choose speed of operation so as to minimize the sum of the flow time and energy. 

In prior work, speed scaling problem with multiple servers has been considered \cite{bansal2009speed,devanur2017primal, vaze2020multiple, vaze2020network}, however, with a fixed number of servers and without any upper bound on the power consumption of any server. 
For a single server, the flow time + energy problem under a power constraint or upper limit on speed has been solved in \cite{bansal2009speed}.
In this paper, we consider the speed scaling problem to minimize the sum of the flow time plus energy with infinite servers under a sum-power constraint across all servers. We refer to this as the {\it flow time + energy} problem. 
Even though there are unlimited number of servers, each job can only be processed by one server at any time. 
A special case of this problem is to minimize just the flow time, called the {\it flow time} problem. 


The flow time problem is also equivalent to the problem of scheduling parallelizable jobs \cite{harchol2021open} with sub-linear speedup (called \textsc{sub-linear-sched} problem) described as follows. 
Let there be $N$ servers with unit speed, and jobs arriving over time with different sizes have to be assigned a set of servers, so as to minimize  the flow time. 
Jobs receive a concave, sub-linear speedup from parallelization:  decreasing marginal benefit from being allocated additional
servers.  In particular, if $k\le N$ is the number of servers assigned to a job, then the resulting speed obtained is $k^{1/\alpha}$ for $\alpha >1$. 
When jobs can be completely parallelizable $\alpha=1$, processing the job with shortest remaining processing time (SRPT) on all servers is known to be optimal.

In this paper, we consider online algorithms (that have only causal job arrival information) for solving the flow time + energy problem. To quantify the performance of an online algorithm, we consider the metric of  competitive ratio, that is defined as the ratio of the flow time of the online algorithm and the optimal offline algorithm $\opt$ maximized over all possible inputs (worst case).

{\bf Prior Work}
The \textsc{sub-linear-sched} problem  has been an object of immense interest \cite{berg2017towards, berg2019hesrpt, im2016competitively, edmonds2000scheduling, edmonds2009scalably, agrawal2016scheduling}, where practical algorithms include packing based \cite{verma2015large}, and resource reservation algorithms \cite{ren2016clairvoyant}. Heuristic policies with only numerical performance analysis  can be found in \cite{lin2018model}.
In past, this problem has been considered for the combinatorial discrete allocation model \cite{im2016competitively}, where an integer number of servers are assigned to any job, as well as the continuous allocation model \cite{edmonds2000scheduling, edmonds2009scalably, agrawal2016scheduling, berg2017towards, berg2019hesrpt}, that treats the $N$ servers as a single resource block which can be partitioned into any size and assigned to any job.

For the discrete allocation model, in \cite{im2016competitively}, a variant of the SRPT algorithm is shown to be $4^{1/(1-1/\alpha)} \log W$ competitive, where $W = w_{\max}/w_{\min}$ is the ratio of the largest and the smallest job size. 
For the continuous allocation model, this problem has been considered in \cite{edmonds2000scheduling, edmonds2009scalably, agrawal2016scheduling, berg2017towards, berg2019hesrpt}, where with resource augmentation, i.e., the online algorithm is allowed more resources, e.g., faster machines, than the $\opt$, algorithms with constant competitive ratios have been derived as a function of the resource augmentation factor.

A special case of the problem in the continuous allocation model has been considered in \cite{berg2019hesrpt} recently, where all jobs arrive together/are available at time $0$. In this simpler setting, \cite{berg2019hesrpt} derived an optimal algorithm, called heSRPT, that gave an explicit expression for the number of servers to be dedicated for each job, which prefers smaller jobs, but unlike SRPT, all jobs are given non-zero speed. In the stochastic setting, where jobs arrive over time that have exponentially distributed job sizes, \cite{berg2017towards} showed that an algorithm called EQUI that dedicates equal number of resources to all outstanding jobs is optimal to minimize the expected flow time.



Our contributions for the flow time problem are as follows.
\begin{itemize}
\item We first consider the setting similar to \cite{berg2019hesrpt}, where all jobs are available at time $0$, and show that the well known algorithm called EQUI, that allocates equal number of resources to all outstanding jobs has competitive ratio of $\left(2-\frac{1}{\alpha}\right) \frac{1}{\left(1-\left(\frac{1}{\alpha}\right)\right)}$. For example, for $\alpha=2$, the bound is at most $3$. 
As $\alpha$ increases, speeds chosen by EQUI and the optimal algorithm (heSRPT \eqref{def:hesrpt} \cite{berg2019hesrpt}) converge, and that is also reflected in the competitive ratio bound that improves as $\alpha$ increases. 
In constrast to heSRPT, the utility of EQUI is that it does not need to know the exact remaining sizes of the jobs, and the speed is identical for all jobs which makes it easy to implement in practice. 
\item For the online setting, where jobs arrive over time, we propose an algorithm called the \textsc{Fractional-LCFS-EQUI} that processes a fraction of the outstanding jobs that have arrived most recently, with equal speed. 
We show that \textsc{Fractional-LCFS-EQUI} has a competitive ratio that depends only on $\alpha$ and not on system parameters such as the total number of jobs, and their respective sizes. In prior work, for similar results, resource augmentation was needed \cite{edmonds2000scheduling, edmonds2009scalably, agrawal2016scheduling}. Thus, our result overcomes a fundamental bottleneck of resource augmentation compared to  \cite{edmonds2000scheduling, edmonds2009scalably, agrawal2016scheduling}, however, it must be noted that  \cite{edmonds2000scheduling, edmonds2009scalably, agrawal2016scheduling} considered more complicated problem setups as discussed earlier.
\item For the more general flow time + energy problem, the competitive ratio bound is at most twice their counterparts for the flow time problem in the two respective cases.
\end{itemize}

\section{Problem Formulation}

Let there be an unlimited supply of servers, however, any job can be processed on any one server at any time (no parallelization is allowed). We consider that preemption is allowed, i.e., a job can be halted at any time and restarted later on any server.
Each server when executing a job at speed $s(t)$ at time $t$, consumes power $P(s(t))$. Let $P(x) = x^\alpha$ where $\alpha > 1$. In addition, there is a sum-power constraint of $\bp$ across all servers, i.e., $\sum_{i : s_i(t) > 0} P(s_i(t)) \le \bp$, for all $t$, where $s_i(t)$ is the speed of an active server $i$ at time $t$. 
 
If a server works with speed $s$ for a time duration $t$ on a job, then $st$ amount of work is completed for that job.
Set of jobs $\cJ$ arrive over time, where an arriving job $j \in \cJ$ with size $w_j$ is defined to be complete at time $d_j$, if $w_j$ amount of work has been completed for it by time $d_j$, where the work could have been done by different servers at different times. Hence the flow time is $\sum_{j\in \cJ} (d_j-a_j) = \int n(t) dt$, where $n(t)$ is the number of unfinished jobs in the system at time $t$.

The flow time problem is then
\begin{align}\label{eq:flowtime}
\min  \int n(t)dt, &
\end{align} subject to  
$ \ \sum_{i : s_i(t) > 0} P(s_i(t)) \le \bp,$ 
while the flow time + energy problem  is
\begin{align}\label{eq:flowtimeplusenergy}
\min  \int n(t)dt  + \int \sum_{i : s_i(t) > 0} P(s_i(t)) dt, &
\end{align}
$\quad \text{subject to} \ \sum_{i : s_i(t) > 0} P(s_i(t)) \le \bp. $
\begin{remark} Note that even under the sum power constraint, the metric of flow time + energy is meaningful, since it is not necessary that the energy used by an optimal algorithm is equal to the maximum possible allowed by the constraint. For example, when the number of outstanding jobs is small, an optimal algorithm may choose a small speed such that the total power consumed is less than the sum-power constraint.
\end{remark}
Next, we show that problem \eqref{eq:flowtime} is equivalent to the \textsc{sub-linear-sched} problem, that  has $N$ parallel and identical servers. 
Similar to \cite{edmonds2000scheduling, berg2019hesrpt}, we consider the continuous allocation model, where $N$ is treated as a single resource block which can be divided into chunks of arbitrary sizes and allocated to different jobs.
Any job is parallelizable with concave speedup, i.e.,  if job~$j$ is allotted 
$k_j(t)$ number of
servers at time $t$, then the service rate experienced by
job~$j$ at time $t$ is $s_j(t) = S(k_j(t)) = k_j(t)^{1/\alpha},$ where $\alpha > 1.$ Here,
$S(\cdot)$ denotes the speedup function, that is concave. The parameter $\alpha$ controls the parallelizability of any job, and depending on $\alpha$, jobs experience appropriate diminishing returns
from being allocated additional servers. Note that
\begin{equation}\label{eq:serverconst}
\sum_{j= 1}^{A(t)} k_j(t) \leq N, 
\end{equation} where $A(t)$ is the set of jobs that are given non-zero service rate at time $t$. Note that the objective is to minimize the flow time of all jobs.


To cast the \textsc{sub-linear-sched} problem as a flow time problem \eqref{eq:flowtime}, suppose that for each job we can `create'
its own dedicated server, and a job is processed on only one server, and cannot be parallelized.
Let $s_j(t)$ denote the speed allocated to the server processing job~$j$ at time $t$.
Let 
$k_j(t) = P(s_j(t)) := S^{-1}(s_j(t))$ be the power consumption of job~$j$ on its own server if it is processed at speed $s_j(t)$, where $P(s) = s^{\alpha}, \alpha>1$. Then, \eqref{eq:serverconst} is equivalent to 
$\sum_{j \in A(t)} P(s_j(t)) \leq N$,  the total
power used across all the active servers is at most $N$.

In both these models, the service speed of job~$j$ is $s_j,$ so the
flow times across both the models would be identical. Letting $N=\bp$, we see that  \textsc{sub-linear-sched} is equivalent to the flow time problem \eqref{eq:flowtime}.


In the following, we will consider Problem \eqref{eq:flowtimeplusenergy}, and propose algorithms and bound their competitive ratios \eqref{defn:cr}. Minimal changes to be made for algorithms to be feasible, and analysis to be applicable for Problem \eqref{eq:flowtime} are mentioned in Remark \ref{rem:prob1}. Consequently, we will only indicate the corresponding competitive ratio results for Problem \eqref{eq:flowtime}.

{\bf Metric}
We represent the optimal offline algorithm (that knows the entire job arrival sequence in advance) as $\opt$. 
Let $n(t)$ ($n_o(t)$) and $P_{\text{sum}}(t)$ ($P^o_{\text{sum}}(t)$) be the number of outstanding jobs with an online algorithm $\cA$ ($\opt$), and the sum of the power used by an online algorithm $\cA$ ($\opt$) across all servers at time $t$, respectively. For Problem \eqref{eq:flowtimeplusenergy}, we will consider the metric of competitive ratio which for an algorithm $\cA$ is defined as 
\begin{equation}\label{defn:cr}\mu_\cA  = \max_\sigma \frac{\int (n(t) + P_{\text{sum}}(t)) dt}{\int (n_o(t) + P^o_{\text{sum}}(t)) dt},
\end{equation}
where $\sigma$ is the input sequence consisting of jobs set $\cJ$. Since $\sigma$ is arbitrary, we are not making any assumption on the job arrival times, or their sizes. 

We will propose an online algorithm $\cA$, and bound $\mu_\cA \le \kappa$, by showing that for each time instant $t$
\begin{align}\label{eq:runcond}
  n(t) + P_{\text{sum}}(t) + d\Phi(t)/dt & \le \kappa( n_o(t) + P^o_{\text{sum}}(t)),
\end{align}
where $\Phi(t)$ is some function called the {\bf potential function} that satisfies the boundary conditions:
\begin{itemize}
\item $\Phi(t) = 0$ initially before all job arrivals and $\Phi(\infty) = 0$.
\item $\Phi(t)$ does not increase on any job arrival or job departure with the algorithm or the $\opt$.
\end{itemize}
Integrating \eqref{eq:runcond} over time, implies that the competitive ratio of $\cA$ is at most $\kappa$.

\section{All jobs available at time $0$} We first consider the simpler setting when all jobs of $\cJ$ arrive together at time $0$. For Problem \eqref{eq:flowtime}, this setting has been considered recently \cite{berg2019hesrpt}, and an optimal algorithm (called heSRPT) has been derived. 
In particular, with heSRPT, let at time $t$ there are $n(t)$ unfinished jobs that are indexed in decreasing order of their remaining sizes, $w_{n(t)}(t) \le \dots \le w_2(t) \le w_1(t)$. Then the number of servers dedicated to job $i=1, \dots, n(t)$ is 
\begin{equation}\label{def:hesrpt}k_i(t)= N \left( \left(\frac{i}{n(t)}\right)^{\left(\frac{1}{1-1/\alpha}\right)} - \left(\frac{i-1}{n(t)}\right)^{\left(\frac{1}{1-1/\alpha}\right)}\right),
\end{equation}
and corresponding speed is $S(k_i(t)) = k_i(t)^{1/\alpha}$.
Thus, shorter jobs get more servers, and consequently more speed.

We show that a simpler algorithm, called EQUI, that processes all jobs simultaneously at the same speed, and does not require the knowledge of remaining job sizes, has a  constant (depending on $\alpha$) competitive ratio for both Problem \eqref{eq:flowtime} and Problem \eqref{eq:flowtimeplusenergy}. 

We begin with some preliminaries.
Let $Q(x) = \frac{x}{P^{-1}(x)}$. For $P(x) = x^\alpha$, $Q(x) = x^{1-\frac{1}{\alpha}}$.
\begin{lemma}\label{lem:optmaxspeedspecial}
With $P^o_{\text{sum}}(t)$ as the sum power used by $\opt$ at any time $t$, the maximum speed devoted to processing any one job by the $\opt$ is at most $P^{-1}(P^o_{\text{sum}}(t))$. 
Moreover,  the sum of the speeds with which $\opt$ is processing any of its $k$ jobs is at most $Q(k)P^{-1}(P^o_{\text{sum}}(t))$.
\end{lemma}
Proof is trivial and hence omitted.

\subsection{Algorithm EQUI}
At time $t$, if the outstanding number of jobs in the system is $n(t)$, then all $n(t)$ jobs are processed parallely on
$n(t)$ servers, each with identical speed 
\begin{equation}\label{defn:speedchoicespecialFT}
s(t) = P^{-1}\left(\frac{\min\{n(t), \bp\}}{n(t)}\right).
\end{equation}
 Thus, the total power 
used by EQUI at any time is $\le  n(t) P\left(P^{-1}\left(\frac{\min\{n(t), \bp\}}{n(t)}\right)\right) \le \min\{n(t), \bp\} \le  \bp$.
\begin{remark}\label{rem:prob1} All algorithms presented in the paper when applied to Problem \eqref{eq:flowtime} will have the term $\min\{n(t), \bp\}$ in their speed choice replaced by $\bp$. Similarly, for potential functions defined in \eqref{defn:phispecialFT}  and \eqref{defn:phi}, terms $\min\{., \bp\}$ will be replaced by $ \bp$.
\end{remark}
\subsection{Potential Function} 
At time $t$, let $A(t)$ be the set of unfinished jobs with EQUI  with $n(t) = |A(t)|$, and for the $i^{th}$ job,  $i\in A(t)$, let $q_i(t)$ be its remaining size. Then  
\begin{align}\label{}
 n^i(t,q) & = \begin{cases} 1 & \text{for} \  q\le q_i(t), \\
 0 & \text{otherwise.}
 \end{cases}
\end{align} 
Similarly, let $n_o(t)$ be the number of unfinished jobs with the $\opt$, and the corresponding quantity to $n^i(t,q)$ for the $i^{th}$ job with the $\opt$, be denoted by $n^i_o(t,q)$.

Consider the potential function  $  \Phi_{sf}(t) =$
\begin{equation}\label{defn:phispecialFT} 
 c_1 P^{-1}\left(\frac{n(t)}{\min\{n(t), \bp\}}\right) \left(\sum_{i\in A(t)}\int_{0}^\infty ( n^i(t,q) - n^i_o(t,q))^+ dq\right),
\end{equation}
where $c_1$ is a constant to be chosen later, and $(x)^+ = \max\{0,x\}$.

Clearly, $\Phi_{sf}(t)$ satisfies the first boundary condition. 
Since all jobs are available at time $0$, to check whether $\Phi_{sf}(t)$ satisfies the second boundary condition, 
we only need to check whether $\Phi_{sf}(t)$ increases on a departure of a job with either the EQUI or the $\opt$.
\begin{lemma}\label{lem:jumpequi}
Potential function $\Phi_{sf}(t)$ \eqref{defn:phispecialFT} does not increase on a departure of a job with either the EQUI or the $\opt$.
\end{lemma}
Next, we characterize the drift $d \Phi_{sf}(t)/dt$.
\begin{lemma}\label{lem:driftphis} $d \Phi_{sf}(t)/dt  \le  -c_1((-\max\{n(t)-n_o(t),0\} ) $
\begin{align*}
  &  \quad  +  c_1 \left(\frac{1}{\alpha}\right) \max\{n(t),P^o_{\text{sum}}(t)\} + c_1\left(1-\frac{1}{\alpha}\right) n_o(t).
\end{align*}
\end{lemma}

\begin{theorem}\label{thm:flowtimeplusenergy}
  The competitive ratio of EQUI for Problem \eqref{eq:flowtimeplusenergy} when all jobs are available at time $0$, is at most $$\mu(\alpha) = \left(2-\frac{1}{\alpha}\right) \frac{2}{\left(1-\left(\frac{1}{\alpha}\right)\right)}.$$ For $\alpha=2$, $\mu(2) = 6$. 
  Moreover, $\mu(\alpha)$ is a decreasing function of $\alpha > 1$.
\end{theorem}
\begin{proof}
Case I : $\max\{n(t)-n_o(t),0\}  =0.$ In this case, from Lemma \ref{lem:driftphis}, we can write   \eqref{eq:runcond}, as $n(t)  + P_{\text{sum}}(t) +  d\Phi(t)/dt$
\begin{align}\nn
& = n(t)   +  \min\{n(t),\bp\} + c_1  \left(\frac{1}{\alpha}\right) \max\{n(t),P^o_{\text{sum}}(t)\} \\ \nn
& \quad + c_1 \left(1-\frac{1}{\alpha}\right) n_o(t), \\ \nn
& \le 2n(t)   +  c_1 \left(\frac{1}{\alpha}\right) (n(t)+P^o_{\text{sum}}(t))+ c_1 \left(1-\frac{1}{\alpha}\right)n_o(t),\\ \nn
  & \stackrel{(a)}\le (2+c_1)n_o(t) + c_1\left(\frac{1}{\alpha}\right) P^o_{\text{sum}}(t), \\\label{eq:runcond1}
  & \le (2+c_1)(n_o(t) +P^o_{\text{sum}}(t)),
\end{align}
where $(a)$ follows since $n(t) \le n_o(t)$.

Case II: $n_o(t)> 0$, and $\max\{n(t)-n_o(t),0\} = n(t)-n_o(t)$.
Using Lemma \ref{lem:driftphis}, we can write \eqref{eq:runcond}, as $n(t)  + P_{\text{sum}}(t)  + d\Phi(t)/dt $
\begin{align}\nn
  & = n(t)  + \min\{n(t),\bp\} - c_1(n(t)-n_o(t)) \\ \nn 
  & \quad  +   c_1\left(\frac{1}{\alpha}\right)  \max\{n(t),P^o_{\text{sum}}(t)\}  + c_1 \left(1-\frac{1}{\alpha}\right) n_o(t), \\  \nn
  & \le 2n(t)  - c_1(n(t)-n_o(t))+   c_1 \left(\frac{1}{\alpha}\right) (n(t)+P^o_{\text{sum}}(t))\\ \nn 
  & \quad + c_1\left(1-\frac{1}{\alpha}\right) n_o(t), \\  \nn
  & \le n(t)\left(2-c_1 + c_1\left(\frac{1}{\alpha}\right)\right) + n_o(t)\left(c_1\left(1+\left(1-\frac{1}{\alpha}\right)\right)\right)\\ \nn
  & \quad  + c_1 \left(\frac{1}{\alpha}\right)P^o_{\text{sum}}(t), \\  
  & \le \left(c_1\left(2-\frac{1}{\alpha}\right)\right) (n_o(t) +P^o_{\text{sum}}(t) )\label{eq:runcond2}
\end{align}
for $ c_1\ge 2/\left(1-\left(\frac{1}{\alpha}\right)\right)$.
When $n_o(t) =0$, then we do not have to add the contribution of the $\opt$ from Lemma \ref{lem:driftphis}, and we get $n(t)  + P_{\text{sum}}(t)  + d\Phi(t)/dt $
\begin{align}\nn
  & = n(t)  + n(t) \left(\frac{\min\{n(t), \bp\}}{n(t)}\right) - c_1n(t), \\\label{eq:runcond3}
  & \le n(t)(2-c_1) \le 0,
\end{align}
for $c_1=2$.
Thus, choosing $c_1= 2/\left(1-\left(\frac{1}{\alpha}\right)\right)$, 
from \eqref{eq:runcond1}, \eqref{eq:runcond2},  and \eqref{eq:runcond3},  \eqref{eq:runcond} holds for $\kappa= 
\left(2-\frac{1}{\alpha}\right) \frac{2}{\left(1-\left(\frac{1}{\alpha}\right)\right)}$.
\end{proof} 

Identical proof follows for Problem \eqref{eq:flowtime}, since essentially, the only difference when considering Problem \eqref{eq:flowtime} is that we can remove the energy term corresponding to $P_{\text{sum}}(t)$, and choose $c_1 = 1/\left(1-\left(\frac{1}{\alpha}\right)\right)$, and show that \eqref{eq:runcond} holds for $\kappa=\left(2-\frac{1}{\alpha}\right) \frac{1}{\left(1-\left(\frac{1}{\alpha}\right)\right)}$.

\begin{theorem}
  The competitive ratio of  EQUI for Problem \eqref{eq:flowtime} when all jobs are available at time $0$, 
  is at most $\left(2-\frac{1}{\alpha}\right) \frac{1}{\left(1-\left(\frac{1}{\alpha}\right)\right)}$. For $\alpha=2$, the upper bound is at most $3$, and decreases to $2$ as $\alpha$ increases.
\end{theorem}

{\it Discussion:} In this section, we showed that for minimizing flow time (Problem \eqref{eq:flowtime}), a simpler algorithm (EQUI) than the optimal heSRPT algorithm \cite{berg2019hesrpt}, that processes all available jobs with the same speed, is constant (depending only on $\alpha$) competitive. Moreover, as $\alpha$ increases, speeds chosen by EQUI and heSRPT algorithm \eqref{def:hesrpt} converge, and that is also reflected in the competitive ratio bound that improves as $\alpha$ increases. 
Thus, knowing job sizes and using job dependent speed is not critical for staying close to the optimal performance.
The utility of EQUI is that it does not need to know the exact remaining sizes of the jobs, thus making it applicable for more wider network setting where pipelining \cite{isard2007dryad, condie2010mapreduce, rossbach2013dandelion} is implemented, and jobs on arrival do not reveal their true sizes. 

\section{Online Job Arrivals}
In this section, we consider Problem \eqref{eq:flowtimeplusenergy} in the online case, where jobs arrive over time with arbitrary sizes and at arbitrary time instants. 
\subsection{Algorithm \textsc{Fractional-LCFS-EQUI}}

At time $t$, let the outstanding number of jobs in the system be $n(t)$.
{\bf Scheduling:} Process the $\beta n(t), \beta<1,$ jobs that have arrived {\bf most recently}  in their respective 
$\beta n(t)$ servers. \footnote{If $\beta n(t)$ is fractional, then we mean $\lceil \beta n(t) \rceil$.} 
{\bf Speed:} Use EQUI, to process all the $\beta n(t)$ jobs  at equal speed 
$s(t) = P^{-1}\left(\frac{\min\{n(t),\bp\}}{ \beta n(t)}\right)$. 

By its very definition, \textsc{Fractional-LCFS-EQUI} satisfies the the total power 
constraint as follows $P_\text{sum}(t)$
$$ \le \beta n(t) P\left(P^{-1}\left(\frac{\min\{n(t),\bp\}}{\beta n(t)}\right)\right) \le \min\{n(t),\bp\}\le  \bp.$$
\begin{remark}Choosing $\beta=1$, \textsc{Fractional-LCFS-EQUI} is identical to EQUI.  Intuitively following heSRPT algorithm \eqref{def:hesrpt} at each time $t$ in the online case appears better than \textsc{Fractional-LCFS-EQUI}, since it is locally optimal, however, analyzing the heSRPT algorithm in the online case appears challenging. 
\end{remark}

The main result of this section is as follows.

\begin{theorem}\label{thm:flowtimeplusenergyonline}
  For any $\alpha > 1$, there exists a $\beta < 1$, such that the competitive ratio of algorithm \textsc{Fractional-LCFS-EQUI} for Problem \eqref{eq:flowtimeplusenergy} is a constant (depends only on $\alpha$) and is independent of the number of jobs, and their sizes. For example, for $2\le \alpha \le 3$, with $\beta = \frac{1}{6}$, the competitive ratio is at most $693$.
\end{theorem}
We get the result for Problem \eqref{eq:flowtime} as a corollary as follows.
\begin{corollary}\label{cor:online}
  For any $\alpha > 1$, there exists a $\beta < 1$, such that the competitive ratio of algorithm \textsc{Fractional-LCFS-EQUI} for Problem \eqref{eq:flowtime} is a constant (depends only on $\alpha$) and is independent of the number of jobs, and their sizes. In particular, the competitive ratio will be at most half of the competitive ratio for Problem \eqref{eq:flowtimeplusenergy}. For example, for $2\le \alpha \le 3$, with $\beta = \frac{1}{6}$, the competitive ratio is at most $345$.
\end{corollary}

{\it Discussion:} Similar to EQUI, algorithm \textsc{Fractional-LCFS-EQUI} is also a non-clairvoyant algorithm, i.e., it does not need to know the remaining size of any outstanding job. The main novelty of Theorem \ref{thm:flowtimeplusenergyonline} over previous such results \cite{edmonds2000scheduling, edmonds2009scalably, gupta2012scheduling}, is that it is proven without needing resource augmentation. With resource augmentation, an online algorithm is given servers that are allowed to operate at speed $s(1+\theta), \theta >0$ while consuming only power $P(s)$, but the $\opt'$s consumption is kept intact at $P(s)$ with speed $s$. 
Thus, an online algorithm is given extra/faster resources. In \cite{edmonds2000scheduling, edmonds2009scalably, gupta2012scheduling}, algorithms with competitive ratio as a function of $\alpha$ and $\theta$ have been derived for a similar but more complicated non-clairvoyant setting. 



\subsection{Proof of Theorem \ref{thm:flowtimeplusenergyonline}} 
From here on we refer to algorithm \textsc{Fractional-LCFS-EQUI} as just algorithm.
Let at time $t$, the set of outstanding (unfinished) number of jobs with the algorithm be $A(t)$ with $n(t) = |A(t)|$. 
Let at time $t$, the {\bf rank} $r_j(t)$ of a job $j \in A(t)$  be equal to the number of outstanding jobs of $A(t)$ with the algorithm that have arrived before job $j$. Note that the rank of a job does not change on arrival of a new job, but can change if a job departs that had arrived earlier.

As before, $Q(x) = \frac{x}{P^{-1}(x)}$. which specializes to  $Q(x) = x^{1-\frac{1}{\alpha}}$ for $P(x) = x^\alpha$.
Then we consider the following potential function 
\begin{equation}\label{defn:phi}
\Phi(t) = c \sum_{j\in A(t)} \frac{r_j(t)\left(w_j^A(t) - w_j^o(t)\right)^+}{P^{-1}(\min\{r_j(t), \bp\})Q(r_j(t))} ,
\end{equation}
where $w_j^A(t)$ ($w_j^o(t)$) is the remaining size of job $j$ with the algorithm ($\opt$) at time $t$, and $c$ is a constant to be chosen later.
\begin{remark}The potential function \eqref{defn:phi} is very similar to the one used in \cite{gupta2012scheduling} for a very different problem. The main novelty of our result is that we avoid resource augmentation unlike \cite{gupta2012scheduling}. Moreover, we would like to point out that for Problem \eqref{eq:flowtimeplusenergy}, the most popular potential functions used in  \cite{bansal2009speed, vaze2020multiple} cannot be used since they need job processing speed to be a function of $P^{-1}(n(t))$ which is not possible because of the sum-power constraint.
\end{remark}

We next show that the potential function $\Phi(t)$ satisfies the second boundary condition. The fact that the first boundary condition is satisfied is trivial.
\begin{lemma}\label{lem:jumplcfs}
Potential function $\Phi(t)$ \eqref{defn:phi} does not change on arrival of any new job. Moreover, on a departure of a job with the algorithm or the $\opt$, the potential function $\Phi(t)$ \eqref{defn:phi} does not increase.
\end{lemma}

We next bound the drift $d\Phi(t)/dt$ because of the processing by the $\opt$,  and the algorithm, respectively. To avoid cumbersome notation, we write $\beta n(t)$ instead of $\lceil \beta n(t)\rceil$ everywhere.

%

\begin{lemma}\label{lem:optphi}
The change in the potential function \eqref{defn:phi} because of the $\opt's$ contribution
 \begin{align} 
 d\Phi(t)/dt & \le \begin{cases} c P^o_{\text{sum}}(t) & \ \text{if} \ n(t) \le \bp, \\
 c n(t) \frac{Q(n_o(t))}{Q(n(t))} & \ \text{if} \ n(t) > \bp.
 \end{cases}
 \end{align}
\end{lemma}

 \begin{lemma}\label{lem:algphi}
 For $\gamma < \beta$, when $n_o(t) \le \gamma n(t)$, the change in the potential function \eqref{defn:phi} because of the algorithm's contribution is
 \begin{align} 
 d\Phi(t)/dt & \le  -  \frac{(1-\beta) (\beta - \gamma) n(t) }{P^{-1}(\beta)}.
 \end{align}
\end{lemma}

\begin{proof}[Proof of Theorem \ref{thm:flowtimeplusenergyonline}]
  To prove the Theorem, we check the running condition \eqref{eq:runcond} for the two cases separately : i) $n_o(t) > \gamma n(t)$ and then ii) $n_o(t) \le \gamma n(t)$, and show that it holds for a constant $\kappa$.

Case i) $n_o(t) > \gamma n(t)$. In this case, we only count the $\opt's$ contribution to $d\Phi(t)/dt$, which is sufficient since the algorithm's contribution to $d\Phi(t)/dt$ is always non-positive. 
When $n(t) > \bp$, from Lemma \ref{lem:optphi}, we have that \eqref{eq:runcond}, $ n(t) + P_{\text{sum}}(t)+ d\Phi(t)/dt $
\begin{align}\nn
 & \le n(t) + \min\{n(t), \bp\} + c n(t) \frac{Q(n_o(t))}{Q(n(t))},   \\ \nn
  &  \stackrel{(a)}\le 2n(t) + c n(t) \frac{Q(b n(t))}{Q(n(t))} \le  2n(t) + c n(t) b^{1-1/\alpha}, \\  \label{eq:optcont1}
  &   \stackrel{(b)}\le  2n(t)  + c n(t) b =  2n(t)  + c n_o(t) \stackrel{(c)}\le  (2/\gamma+c ) n_o(t), \end{align}
where in $(a)$ we let $n_o(t) = b n(t)$. We first consider the case when $b>1$, where  inequality $(b)$ follows when $b>1$. 
Finally $(c)$ follows since $n_o(t) > \gamma n(t)$.
When $b<1$, then $\frac{Q(b n(t))}{Q(n(t))} < 1$. Thus, similar to \eqref{eq:optcont1}, for $b<1$, we get $n(t) + P_{\text{sum}}(t)+ d\Phi(t)/dt$
\begin{align}\nn
   & \le n(t) + \min\{n(t), \bp\} + c n(t) \frac{Q(n_o(t))}{Q(n(t))}, \\  \nn
    & \le  n(t)(2+c), \\ \label{eq:optcont2}
    & \le \frac{1}{\gamma}(2+c) n_o(t). 
\end{align}
When $n(t) \le \bp$, from Lemma \ref{lem:optphi}, we have that \eqref{eq:runcond}, $ n(t) + P_{\text{sum}}(t)+ d\Phi(t)/dt$
\begin{align} \nn
  & \le n(t) + \min\{n(t), \bp\}  + cP^o_{\text{sum}}(t), \\ \label{eq:optcont2a}
  & \le \frac{2+c}{\gamma}(n_o(t) + P^o_{\text{sum}}(t)). 
  \end{align}

Combining, \eqref{eq:optcont2} and \eqref{eq:optcont2a}, when $n_o(t) > \gamma n(t)$ 
\begin{align}\label{eq:optcont3}
  n(t) + P_{\text{sum}}(t)+ d\Phi(t)/dt, 
    & \le \frac{1}{\gamma}(2+c) (n_o(t) + P^o_{\text{sum}}(t)) .
\end{align}

Case ii) $n_o(t) \le \gamma n(t)$. 
Let $n_o(t)>0$ .
When $n(t) > \bp$, From Lemma \ref{lem:optphi} and Lemma \ref{lem:algphi}, \eqref{eq:runcond} can be bounded as  $n(t) + P_{\text{sum}}(t) + d\Phi(t)/dt   $
\begin{align}\nn
& \le n(t) + \min\{n(t),\bp\}  \\ 
& \quad +   c n(t) \frac{Q(n_o(t))}{Q(n(t))} -  c \frac{(1-\beta) (\beta - \gamma) n(t) }{P^{-1}(\beta)},\\\label{eq:c2}
  &  \stackrel{(a)}\le  n(t) \left(2+ c\left(\gamma^{1-1/\alpha} - \frac{(1-\beta) (\beta - \gamma) }{P^{-1}(\beta)}\right)\right) \stackrel{(b)} \le  0,
  \end{align}
  where $(a)$ follows since $n_o(t)\le \gamma n(t)$, while $(b)$ follows
 for choice of $\gamma, \beta, c$ that satisfy
   \begin{equation}\label{eq:condsecc}
 \frac{(1-\beta) (\beta - \gamma) }{P^{-1}(\beta)} > \gamma^{1-1/\alpha} \ \text{and} \  c \ge  \frac{-2}{\left(\gamma^{1-1/\alpha} - \frac{(1-\beta) (\beta - \gamma) }{P^{-1}(\beta)}\right)}.
\end{equation}
When $n_o(t)=0$, the $\opt$'s contribution is zero, and \eqref{eq:c2} holds with for a smaller value of $c$.
  
Similarly, when $n(t) \le \bp$, $n(t) + P_{\text{sum}}(t) + d\Phi(t)/dt  $
\begin{align}\nn
& \le n(t) + \min\{n(t),\bp\} +  c P^o_{\text{sum}}(t) -  c \frac{(1-\beta) (\beta - \gamma) n(t) }{P^{-1}(\beta)},\\\nn
  &  \le  2n(t) \left(1- c\left( \frac{(1-\beta) (\beta - \gamma) }{P^{-1}(\beta)}\right)\right)+ c P^o_{\text{sum}}(t), \\ \label{eq:c3}
  &\stackrel{(a)} \le cP^o_{\text{sum}}(t),
  \end{align}
where $(a)$ follows as long as $c\left( \frac{(1-\beta) (\beta - \gamma) }{P^{-1}(\beta)}\right) > 1$. 
Combining \eqref{eq:optcont3}, \eqref{eq:c2} and \eqref{eq:c3}, using \eqref{eq:runcond}, the competitive ratio of 
the proposed algorithm  is 
\begin{equation}\label{eq:finalbound}
\frac{2+c}{\gamma}
\end{equation} where $c, \beta$ and $\gamma$ satisfy \eqref{eq:condsecc}. 
Note that the bound \eqref{eq:finalbound} can be optimized by choosing the optimal value of $\beta$ and $\gamma$ satisfying \eqref{eq:condsecc}. 
It is easy to see that depending on $\alpha$, there exists a $\beta$ satisfying \eqref{eq:condsecc} with $\gamma = \beta^2$.  

For example, for $\alpha=2,3$, let $\beta = \frac{1}{6}$ and $\gamma = \beta^2$,  and we get a competitive ratio bound of $693$ and $680$, respectively, as follows. In fact for $2\le \alpha \le 3$, choosing $\beta = \frac{1}{6}$ and $\gamma = \beta^2$, the competitive ratio is at most $693$. 
For $\alpha =2$, let $\beta = \frac{1}{6}$, and $\gamma = \beta^2$. 
Then $\frac{(1-\beta) (\beta - \gamma) }{P^{-1}(\beta)} = 5/6(5/36)/\sqrt{1/6} = 2.45  5/6(5/36) =.283$ while $\gamma^{1-1/\alpha}= .167$. Thus, we have $\frac{(1-\beta) (\beta - \gamma) }{P^{-1}(\beta)} > \gamma^{1-1/\alpha}$ and 
$c =2/ (.283-.167) = 17.24$.
Thus, the competitive ratio when $\alpha=2$ is $\frac{2+c}{\gamma}= 36\times(2+17.24)< 693$.

Similarly for $\alpha=3$, with $\beta = \frac{1}{6}$, and $\gamma = \beta^2$, $c=16.66$ and the competitive ratio upper bound is $\frac{2+c}{\gamma}=36\times(18.66)<680$.
\end{proof}

\begin{proof}[Proof of Corollary \ref{cor:online}] Proof is immediate by noting that with Problem \eqref{eq:flowtime}, we do not have to add the energy consumption term $P_{\text{sum}}(t)$ for checking the running condition \eqref{eq:runcond}, thus resulting in a two-fold decrease in the $\kappa$ needed to satisfy \eqref{eq:runcond}.
\end{proof}

\section{Conclusions}
In this paper, we considered an important problem of flow time minimization in  data centers, where jobs have limited parallelizability, and 
they experience diminishing returns from being allocated additional servers. When all jobs are available at time $0$, a very simple algorithm called EQUI that processes all outstanding jobs at the same speed is shown to have a  constant competitive ratio that only depends on the speed-up exponent $\alpha$. For the most relevant speed-up exponents of $2$ and $3$, the competitive ratio is at most $3$. Thus, even without knowing job-sizes, and processing all of them at the same speed, EQUI is not too sub-optimal.

For the general online setting, where jobs arrive over time, we propose a LCFS type algorithm for scheduling and EQUI for speed selection, and show that its competitive ratio is a constant that only depends on the speedup exponent $\alpha$. Our result overcomes fundamental difficulty found in literature where similar results were shown only in the presence of resource augmentation, where an online algorithm is allowed more resources than the optimal offline algorithm. 
\bibliography{refs}

\begin{thebibliography}{10}

\bibitem{agrawal2016scheduling}
K.~Agrawal, J.~Li, K.~Lu, and B.~Moseley.
\newblock Scheduling parallelizable jobs online to minimize the maximum flow
  time.
\newblock In {\em Proceedings of the 28th ACM Symposium on Parallelism in
  Algorithms and Architectures}, pages 195--205, 2016.

\bibitem{bansal2009speed}
N.~Bansal, H.-L. Chan, and K.~Pruhs.
\newblock Speed scaling with an arbitrary power function.
\newblock In {\em Proceedings of the twentieth annual ACM-SIAM symposium on
  discrete algorithms}, pages 693--701. SIAM, 2009.

\bibitem{berg2017towards}
B.~Berg, J.-P. Dorsman, and M.~Harchol-Balter.
\newblock Towards optimality in parallel scheduling.
\newblock {\em Proceedings of the ACM on Measurement and Analysis of Computing
  Systems}, 1(2):1--30, 2017.

\bibitem{berg2019hesrpt}
B.~Berg, R.~Vesilo, and M.~Harchol{-}Balter.
\newblock {heSRPT}: Optimal scheduling of parallel jobs with known sizes.
\newblock {\em {SIGMETRICS} Perform. Evaluation Rev.}, 47(2):18--20, 2019.

\bibitem{condie2010mapreduce}
T.~Condie, N.~Conway, P.~Alvaro, J.~M. Hellerstein, K.~Elmeleegy, and R.~Sears.
\newblock Mapreduce online.
\newblock In {\em Nsdi}, volume~10, page~20, 2010.

\bibitem{devanur2017primal}
N.~R. Devanur and Z.~Huang.
\newblock Primal dual gives almost optimal energy-efficient online algorithms.
\newblock {\em ACM Transactions on Algorithms (TALG)}, 14(1):1--30, 2017.

\bibitem{edmonds2000scheduling}
J.~Edmonds.
\newblock Scheduling in the dark.
\newblock {\em Theoretical Computer Science}, 235(1):109--141, 2000.

\bibitem{edmonds2009scalably}
J.~Edmonds and K.~Pruhs.
\newblock Scalably scheduling processes with arbitrary speedup curves.
\newblock In {\em Proceedings of the twentieth annual ACM-SIAM symposium on
  Discrete algorithms}, pages 685--692. SIAM, 2009.

\bibitem{gupta2012scheduling}
A.~Gupta, S.~Im, R.~Krishnaswamy, B.~Moseley, and K.~Pruhs.
\newblock Scheduling heterogeneous processors isn't as easy as you think.
\newblock In {\em Proceedings of the twenty-third annual ACM-SIAM symposium on
  Discrete algorithms}, pages 1242--1253. SIAM, 2012.

\bibitem{harchol2021open}
M.~Harchol-Balter.
\newblock Open problems in queueing theory inspired by datacenter computing.
\newblock {\em Queueing Systems}, 97(1):3--37, 2021.

\bibitem{im2016competitively}
S.~Im, B.~Moseley, K.~Pruhs, and E.~Torng.
\newblock Competitively scheduling tasks with intermediate parallelizability.
\newblock {\em ACM Transactions on Parallel Computing (TOPC)}, 3(1):1--19,
  2016.

\bibitem{isard2007dryad}
M.~Isard, M.~Budiu, Y.~Yu, A.~Birrell, and D.~Fetterly.
\newblock Dryad: distributed data-parallel programs from sequential building
  blocks.
\newblock In {\em ACM SIGOPS operating systems review}, volume~41, pages
  59--72. ACM, 2007.

\bibitem{lin2018model}
S.-H. Lin, M.~Paolieri, C.-F. Chou, and L.~Golubchik.
\newblock A model-based approach to streamlining distributed training for
  asynchronous sgd.
\newblock In {\em 2018 IEEE 26th International Symposium on Modeling, Analysis,
  and Simulation of Computer and Telecommunication Systems (MASCOTS)}, pages
  306--318. IEEE, 2018.

\bibitem{ren2016clairvoyant}
R.~Ren and X.~Tang.
\newblock Clairvoyant dynamic bin packing for job scheduling with minimum
  server usage time.
\newblock In {\em Proceedings of the 28th ACM Symposium on Parallelism in
  Algorithms and Architectures}, pages 227--237, 2016.

\bibitem{rossbach2013dandelion}
C.~J. Rossbach, Y.~Yu, J.~Currey, J.-P. Martin, and D.~Fetterly.
\newblock Dandelion: a compiler and runtime for heterogeneous systems.
\newblock In {\em Proceedings of the Twenty-Fourth ACM Symposium on Operating
  Systems Principles}, pages 49--68. ACM, 2013.

\bibitem{vaze2020multiple}
R.~Vaze and J.~Nair.
\newblock Multiple server {SRPT} with speed scaling is competitive.
\newblock {\em IEEE/ACM Transactions on Networking}, 28(4):1739--1751, 2020.

\bibitem{vaze2020network}
R.~Vaze and J.~Nair.
\newblock Network speed scaling.
\newblock {\em Performance Evaluation}, 144:102145, 2020.

\bibitem{verma2015large}
A.~Verma, L.~Pedrosa, M.~Korupolu, D.~Oppenheimer, E.~Tune, and J.~Wilkes.
\newblock Large-scale cluster management at google with borg.
\newblock In {\em Proceedings of the Tenth European Conference on Computer
  Systems}, pages 1--17, 2015.

\end{thebibliography}
  \appendix
\section{}
\begin{proof}[Proof of Lemma \ref{lem:jumpequi}]
 On a departure of a job with the algorithm or the $\opt$, $n^i(t,q)$ or $n_o^i(t,q)$ changes for only $q=0$, and since there is an integral outside,  $\int_{0}^\infty ( n^i(t,q) - n_o^i(t,q))^+ dq$ remains the same on a departure of a job with either the algorithm or the $\opt$. 

The pre-factor term $P^{-1}\left(\frac{n(t)}{\min\{n(t), \bp\}}\right)$ changes though, when  there is a departure of a job with the algorithm, on account of $n(t) \rightarrow n(t)-1$. Consider time $t^-$, just before a departure at time $t$, where $n(t) =n(t^-)-1$. 
 
 Case Ia: $\min\{n(t^-), \bp\} = \bp$ and $\min\{n(t), \bp\} = \bp$. In this case,  $P^{-1}\left(\frac{n(t^-)}{\min\{n(t^-), \bp\}}\right) - P^{-1}\left(\frac{n(t)}{\min\{n(t), \bp\}}\right) > 0$. 
 
 Case Ib: $\min\{n(t^-), \bp\} = \bp$ and $\min\{n(t), \bp\} = n(t)$. In this case,  $P^{-1}\left(\frac{n(t^-)}{\min\{n(t^-), \bp\}}\right) - P^{-1}(1) \ge 0$
 
 Case II: $\min\{n(t^-), \bp\} = n(t^-)$ In this case, $$P^{-1}\left(\frac{n(t^-)}{\min\{n(t^-), \bp\}}\right) - P^{-1}\left(\frac{n(t)}{\min\{n(t), \bp\}}\right) = 0.$$
 Thus, the pre-factor $P^{-1}\left(\frac{n(t)}{\min\{n(t), \bp\}}\right)$ does not increase at any job departure. 
 Moreover, the departure of any job with the $\opt$ does not change the pre-factor.
 Since the integral is always non-negative, the assertion of the Lemma follows.
\end{proof}

\section{}
\begin{proof}[Proof of Lemma \ref{lem:jumplcfs}]
 On an arrival of a new job $j$, the ranks of all the existing jobs do not change, while for the newly arrived job $j$, $w_j^A(t) - w_j^o(t)=0$.
  Hence the potential function $\Phi(t)$ \eqref{defn:phi} does not change on arrival of any new job. 
  
 On a departure of a job with the algorithm, rank of any remaining job can only decrease, in particular by $1$. 
 Thus, if at time $t$ when job $k$ departs with the algorithm, job $j$'s ($j\in A(t^{+})$) rank at time  $t^{+}$, is either $r_j(t^+) =  r_j(t)$ or 
 $r_j(t^+) =  r_j(t)-1$. In the first case, there is no change in the potential function. 
 
 In the second case, when $r_j(t^+) =  r_j(t)-1$,
 first consider the case that $r_j(t) < \bp$ which implies that $r_j(t^+) < \bp$. 
 In this subcase $$\frac{r_j(t^+)}{P^{-1}(\min\{r_j(t^+), \bp\})Q(r_j(t^+))} = \frac{r_j(t)}{P^{-1}(\min\{r_j, \bp\})Q(r_j(t))} =1.$$ 
 Thus, there is no change in $\Phi(t)$ in this subcase. 
 
 If instead $r_j(t) \ge \bp$ but $r_j(t^+) = r_j(t)- 1< \bp$ , then 
 $$\frac{r_j(t^+)}{P^{-1}(\min\{r_j(t^+), \bp\})Q(r_j(t^+))} =1,$$ while $\frac{r_j(t)}{P^{-1}(\min\{r_j, \bp\})Q(r_j(t))} =P^{-1}(r_j(t)/\bp) \ge 1$, since $\bp \le r_j(t)$. Thus, in this case also, the potential function $\Phi$ does not increase on departure of a job with the algorithm. Similar conclusion follows if $r_j(t) > \bp$ and $r_j(t^+) > \bp$.
 
Since there is no discontinuity when  a job departs with the $\opt$, hence $\Phi(t)$ does not change when   a job departs with the $\opt$.
\end{proof}

\begin{proof}[Proof of Lemma \ref{lem:optphi}]
From the definition of $\Phi(t)$ \eqref{defn:phi},  $\opt$ can increase $\Phi(t)$ at time $t$ only if it processes jobs that also belong to the set $A(t)$ (outstanding jobs with the algorithm).

From Lemma \ref{lem:optmaxspeedspecial}, with $P^o_{\text{sum}}(t)$ being the total power used by $\opt$ at time $t$, the sum of the speeds devoted to the  set of $A(t)$ jobs of the algorithms with $n(t) = |A(t)|$ by the $\opt$ is at most  
$Q(n(t)) P^{-1}(P^o_{\text{sum}}(t))$. Moreover, since $\opt$ contains only $n_o(t)$ jobs, sum of the speeds devoted to the 
$n(t)$ jobs of the algorithm is at most $$Q(\min\{n(t), n_o(t)\}) P^{-1}(P^o_{\text{sum}}(t)).$$ 

Moreover, from the definition 
of  $\Phi(t)$ \eqref{defn:phi}, the maximum increase in $\Phi(t)$ is possible if the total speed of the $\opt$ that it can dedicate to jobs belonging to $A(t)$ is dedicated to the single job with the largest rank among $A(t)$, i.e., the job with rank equal to $n(t)$. Thus, because of processing by the $\opt$, $  d\Phi(t)/dt  $
\begin{align}\nn
&\le  c  \frac{n(t)}{P^{-1}(\min\{n(t), \bp\}) Q(n(t))} \\ \nn
& \quad  \times Q(\min\{n(t), n_o(t)\}) P^{-1}(P^o_{\text{sum}}(t)).
   \end{align}
   If $n(t) \le \bp$, then 
   \begin{align}\nn
  d\Phi(t)/dt & \le  c Q(n(t))P^{-1}(P^o_{\text{sum}}(t)), \\ \nn
  & \le c Q(\bp)P^{-1}(P^o_{\text{sum}}(t)),\\ \nn
  & \le c Q(\bp)P^{-1}(\bp),\\
  & =c \bp.
   \end{align}
    Otherwise, if $n(t) > \bp$, then 
   \begin{align}\label{}
  d\Phi(t)/dt & \le  c n(t) \frac{Q(n_o(t))}{Q(n(t))},
   \end{align}
   since $P^o_{\text{sum}}(t)\le \bp$.
   
   \end{proof}
   \section{}
\begin{proof}[Proof of Lemma \ref{lem:driftphis}]
Note that for at least $\max\{n(t)-n_o(t),0\}$ jobs belonging to $A(t)$, the corresponding terms $( n^i(t,q) - n_o^i(t,q))^+ > 0$ in $\Phi_{sf}(t)$. Thus, algorithm EQUI is decreasing work at speed $s_i(t)$ for at least $\max\{n(t)-n_o(t),0\}$ jobs. Hence, the drift $d \Phi_{sf}(t)/dt$ with respect to processing by the algorithm EQUI is $d \Phi_{sf}(t)/dt $
\begin{align}\nn
& \stackrel{(a)}= c_1P^{-1}\left(\frac{n(t)}{\min\{n(t), \bp\}}\right)\left(-\max\{n(t)-n_o(t),0\} s_i(t)  \right),\\   \label{eq:phiboundsf1}
  &\stackrel{(b)} = c_1 \left(-\max\{n(t)-n_o(t),0\}  \right)
\end{align}
where $(a)$ follows since for at least $\max\{n(t)-n_o(t),0\}$ jobs, the EQUI algorithm is decreasing work at speed $s_i(t)$, while $(b)$ follows since 
$s_i(t)=P^{-1}\left(\frac{\min\{n(t), \bp\}}{n(t)}\right)$ for all jobs $i$ being processed by the EQUI algorithm. 

From Lemma \ref{lem:optmaxspeedspecial}, we know that the sum of the speeds used by the $\opt$ over its $n_o(t)$ jobs is at most 
\begin{equation}\label{eq:dummy400}
\sum_{i=1}^{n_o(t)} s^o_i(t) \le Q(n_o(t)) P^{-1}(P^o_{\text{sum}}(t)).
\end{equation}

Using this, we bound the $d \Phi_{sf}(t)/dt$ with respect to processing by the $\opt$ as follows
\begin{align} \nn
d \Phi_{sf}(t)/dt &\stackrel{(a)} \le c_1P^{-1}\left(\frac{n(t)}{\min\{n(t), \bp\}}\right)\left(\sum_{i=1}^{n_o(t)} s^o_i(t)\right),\\    \nn
  &\stackrel{(b)} \le c_1P^{-1}\left(\frac{n(t)}{\min\{n(t), \bp\}}\right)Q(n_o(t)) P^{-1}(P^o_{\text{sum}}(t)),   \\ \nn
  & \stackrel{(c)}\le   c_1 P^{-1}(n(t))Q(n_o(t)), \\   \nn
  & =   c_1 n(t)^{1/\alpha} n_o(t) ^{1-1/\alpha}, \\    \label{eq:phiboundsf2}
  & \stackrel{(d)}\le  c_1 \left(\frac{1}{\alpha}\right) n(t) + c_1\left(1-\frac{1}{\alpha}\right) n_o(t). 
\end{align}
where for $(a)$ we assume that all the $n_o$ jobs of the $\opt$ are getting processed at non-zero speed (best case in terms of increasing $d \Phi_{sf}(t)/dt$), while $(b)$ follows from \eqref{eq:dummy400}, $(c)$ holds when $\min\{n(t), \bp\} = \bp$ and since $P^o_{\text{sum}}(t)\le \bp$, and 
$(d)$ 
follows from the generalized AM-GM inequality. \footnote{For $a_i\ge 0$ and $\lambda_i\ge 0$ with $\sum_{i=1}^n \lambda_i=1$, then 
$\prod_{i=1}^n c_i^{\lambda_i} \le \sum_{i=1}^n \lambda_i c_i$.} 

For the case, when $\min\{n(t), \bp\} = n(t)$, following similar steps to reach \eqref{eq:phiboundsf2}, we get 
\begin{align} \nn
d \Phi_{sf}(t)/dt  & \le c_1 P^{-1}(P^o_{\text{sum}}(t))Q(n_o(t)), \\\label{eq:phiboundsf3}
&  \le  
c_1 \left(\frac{1}{\alpha}\right) P^o_{\text{sum}}(t) + c_1\left(1-\frac{1}{\alpha}\right) n_o(t). 
\end{align}

Combining \eqref{eq:phiboundsf1}, \eqref{eq:phiboundsf2}, and \eqref{eq:phiboundsf3}, the proof is complete.
\end{proof}

\section{}
\begin{proof}[Proof of Lemma \ref{lem:algphi}]
  Since the algorithm executes the $\beta n(t)$ jobs that have arrived most recently, the rank of job $i$ that is being processed by the algorithm is 
  $r_i(t) = n(t) - i+1$ for $i=1, \dots, \beta n(t)$. Since $n_o(t) \le \gamma n(t)$, and $\gamma < \beta$,  $$w_j^A(t) - w_j^o(t) > 0$$ for at  least 
  $(\beta - \gamma) n(t)$ jobs with the  algorithm. In the worst case, the ranks of these $(\beta - \gamma) n(t)$ jobs are 
  $(1-\beta) n(t) + i-1$ for $i=1, \dots, (\beta-\gamma) n(t)$.
  
  Since the speed for any of the job executed by the algorithm is $s(t) = P^{-1}\left(\frac{\min\{n(t), \bp\}}{\beta n(t)}\right)$. Thus, the change in the potential function because of the algorithm's  $ d\Phi(t)/dt $
  processing is 
  \begin{align*}\label{}
\le  &  - c\sum_{i= (1-\beta) n(t)}^{(1-\beta) n(t) + (\beta - \gamma) n(t)} \frac{r_i(t)}{P^{-1}(\min\{r_i(t), \bp\})Q(r_i(t))} \\ \nn
& \quad P^{-1}\left(\frac{\min\{n(t), \bp\}}{\beta n(t)}\right),\\
  \stackrel{(a)} \le &   - c\sum_{i= (1-\beta) n(t)}^{(1-\beta) n(t) + (\beta - \gamma) n(t)} \frac{r_i(t)}{Q(n(t))}  \frac{1}{ P^{-1}(\beta n(t))}, \\
  = &   - c\sum_{i= (1-\beta) n(t)}^{(1-\beta) n(t) + (\beta - \gamma) n(t)} \frac{r_i(t)}{Q(n(t))}  \frac{1}{ P^{-1}(n(t))} \frac{1}{P^{-1}(\beta)} , \\
  \stackrel{(b)} = &   - c\sum_{i= (1-\beta) n(t)}^{(1-\beta) n(t) + (\beta - \gamma) n(t)}  \frac{r_i(t)}{n(t)} \frac{1}{P^{-1}(\beta)}, \\
 \stackrel{(c)}\le &  - c \frac{(1-\beta) (\beta - \gamma) n(t) n(t)}{\beta } \frac{1}{n(t) P^{-1}(\beta)} , \\
  = & - c \frac{(1-\beta) (\beta - \gamma) n(t) }{P^{-1}(\beta)} , 
\end{align*}
where $(a)$ follows since $r_i(t) \le n(t)$ and $$\frac{P^{-1}\left(\frac{\min\{n(t), \bp\}}{\beta n(t)}\right)}{P^{-1}(\min\{r_i(t), \bp\})} \ge \frac{1}{ P^{-1}(\beta n(t))},$$ while $(b)$  follows since $Q(x) P^{-1}(x) = x$, and finally $(c)$ follows since there are $(\beta-\gamma) n(t)$ jobs that are being executed each with rank at least $(1-\beta) n(t)$.
\end{proof}

%
%
%
%
\section{Numerical results}\label{sec:sim}
In this section, we present simulation results for the mean flow time (per job). We first consider the special case when all jobs are available at time $0$ and compare the performance of the optimal algorithm, heSRPT, and the EQUI, in Fig. \ref{fig:offline} for different values of $\alpha$.  We use the number of servers $N=1000$, and consider $1000$ jobs with size that is exponentially distributed with mean $20$. For all the results, we compare the two algorithms for the same realization of random variables, and then average it out. As predicted by our results, the performance of EQUI improves compared to heSRPT as $\alpha$ increases. This has also been observed in \cite{berg2019hesrpt}.

Next, for the online setting, we compare the performance of the  proposed algorithm with other known algorithms such as heSRPT and EQUI. 
For all the plots in the online setting, we use the number of servers $N=1000$, and  consider a slotted time where in each slot, the number of jobs arriving is Poisson distributed with mean $20$. For each job, its size is exponentially distributed with mean $20$. For each iteration, we generate jobs for $1000$ slots, and count its flow time, and iterate over 1000 iterations. We plot the results for $\alpha=2, 2.5,$ and $3$. For simplicity of illustration in all the figures we denote the proposed algorithm \textsc{Fractional-LCFS-EQUI} as just Alg with appropriate choice of $\beta$. 
To implement algorithm heSRPT in the online setting, we apply \eqref{def:hesrpt} at each time $t$. For all the results, we compare the performance of different algorithms for the same realization of random variables, and then average it out. 

From all the plots, we can observe that the performance of algorithm \textsc{Fractional-LCFS-EQUI} degrades with decreasing value of $\beta$, but for analytical results sufficiently small value of $\beta$ is needed, for example $\beta=1/6$ in Theorem \ref{thm:flowtimeplusenergyonline} for $\alpha =2$. 
Moreover, we also see that the heSRPT algorithm outperforms algorithm EQUI as well as the proposed algorithm for all choices of $\beta$, however, our analytical guarantees on the competitive ratio are for the proposed algorithm, and upper bounding the competitive ratio of the heSRPT algorithm in the online case is challenging, and is an object of ongoing work.

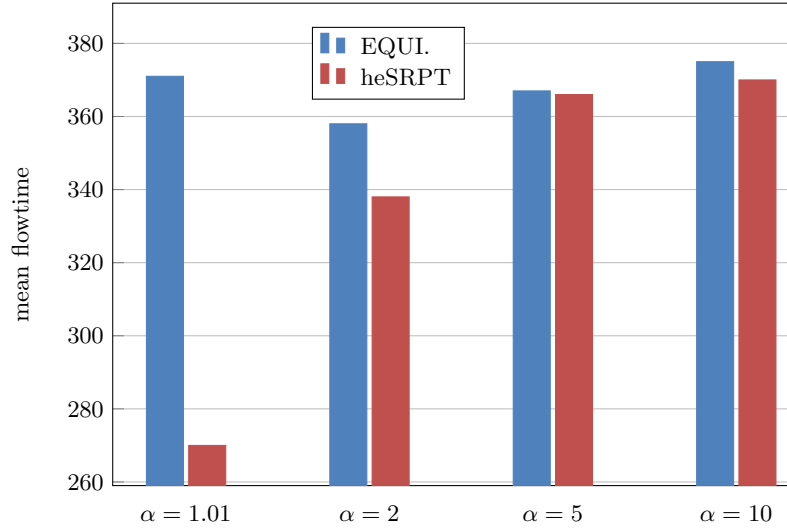
\begin{figure}
\centering
\begin{tikzpicture}
    \begin{axis}[
        width  = 0.75*\textwidth,
        height = 8cm,
        major x tick style = transparent,
        ybar,
        bar width=14pt,
        ymajorgrids = true,
        ylabel = {mean flowtime},
        symbolic x coords={$\alpha=1.01$, $\alpha=2$, $\alpha=5$, $\alpha=10$,$\alpha=20$},
        xtick = data,
        scaled y ticks = false,
        legend cell align=left,
        legend style={
                at={(.4,.8)},
                anchor=south east,
                column sep=1ex}
    ]
        \addplot[style={bblue,fill=bblue,mark=none}]
            coordinates {($\alpha=1.01$, 371) ($\alpha=2$,358) ($\alpha=5$,367) ($\alpha=10$,375) ($\alpha=20$,380)};

        \addplot[style={rred,fill=rred,mark=none}]
            coordinates {($\alpha=1.01$, 270) ($\alpha=2$, 338) ($\alpha=5$,366) ($\alpha=10$, 370) ($\alpha=20$,375) };

%

        \legend{EQUI., heSRPT}
    \end{axis}
\end{tikzpicture}
\caption{Comparison of mean flow time with EQUI and heSRPT.}
\label{fig:offline}
\end{figure}

\begin{figure}
\centering
\begin{tikzpicture}
\begin{axis}[ylabel= mean flowtime,
    symbolic x coords={heSRPT, EQUI, Alg $\beta=\frac{1}{2}$, Alg $\beta=\frac{1}{4}$, Alg $\beta=\frac{1}{6}$},
    xtick=data]
    \addplot[ybar,fill=blue] coordinates {
        (heSRPT,475)
        (EQUI,505)
        (Alg $\beta=\frac{1}{2}$,730)
         (Alg $\beta=\frac{1}{4}$,1043)
          (Alg $\beta=\frac{1}{6}$,1287)
    };
\end{axis}
\end{tikzpicture}
\label{fig:sim2}
\caption{Comparison of flow time for different algorithms with $\alpha=2$.}
\end{figure}
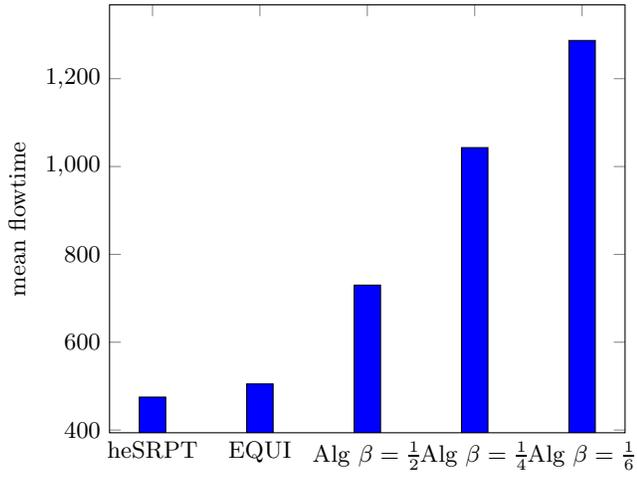

\begin{figure}
\centering
\begin{tikzpicture}
\begin{axis}[ylabel= mean flowtime,
    symbolic x coords={heSRPT, EQUI, Alg $\beta=\frac{1}{2}$, Alg $\beta=\frac{1}{4}$, Alg $\beta=\frac{1}{6}$},
    xtick=data]
    \addplot[ybar,fill=blue] coordinates {
        (heSRPT,462.2)
        (EQUI,476.28)
        (Alg $\beta=\frac{1}{2}$,737)
         (Alg $\beta=\frac{1}{4}$,1127)
          (Alg $\beta=\frac{1}{6}$,1445)
    };
\end{axis}
\end{tikzpicture}
\label{fig:sim25}
\caption{Comparison of flow time for different algorithms with $\alpha=2.5$.}
\end{figure}

\begin{figure}
\centering
\begin{tikzpicture}
\begin{axis}[ylabel= mean flowtime, 
    symbolic x coords={heSRPT, EQUI, Alg $\beta=\frac{1}{2}$, Alg $\beta=\frac{1}{4}$, Alg $\beta=\frac{1}{6}$},
    xtick=data]
    \addplot[ybar,fill=blue] coordinates {
        (heSRPT,451)
        (EQUI,460)
        (Alg $\beta=\frac{1}{2}$,742)
         (Alg $\beta=\frac{1}{4}$,1184)
          (Alg $\beta=\frac{1}{6}$,1562)
    };
\end{axis}
\end{tikzpicture}
\label{fig:sim3}
\caption{Comparison of flow time for different algorithms with $\alpha=3$.}

\end{figure}
\end{document}